\documentclass[11pt]{amsart}
\usepackage{amsmath}
\usepackage{amssymb}
\usepackage{amsfonts}
\usepackage{latexsym}
\usepackage{color}
\usepackage{graphicx}

\catcode `\@=11 \@addtoreset{equation}{section}

\catcode `\@=12

\newcommand{\be}{\begin{equation}}
\newcommand{\en}{\end{equation}}
\newcommand{\bea}{\begin{eqnarray}}
\newcommand{\ena}{\end{eqnarray}}
\newcommand{\beano}{\begin{eqnarray*}}
\newcommand{\enano}{\end{eqnarray*}}
\newcommand{\bee}{\begin{enumerate}}
\newcommand{\ene}{\end{enumerate}}
\newcommand{\ad}{^{\mbox{\scriptsize $\dag$}}}
\newcommand{\mult}{\,{\scriptstyle \square}\,}

\newcommand{\mc}{\mathcal}
\newcommand{\mb}{\mathbb}

\newcommand{\N}{\mathbb N}

\newcommand{\w}{{\rm w}}
\def\H{{\mathcal H}}
\newcommand{\Hil}{{\mathcal H}}

\newcommand{\Id}{1\!\!1}

\def\L{{\mathcal L}}
\newcommand{\Lc}{{\mathcal L}}
\newcommand{\LL}{{\mathcal L}}

\newcommand{\D}{{\mathcal D}}

\newcommand{\E}{{\mathcal E}}

\newcommand{\M}{{\mathfrak M}}
\newcommand{\A}{{\mathfrak A}}
\newcommand{\Ao}{{\mathfrak A}_0}
\newcommand{\up}{\upharpoonright}
\newcommand{\restr}[1]{\!\up\!{#1}}

\newtheorem{thm}{Theorem}[section]
\newtheorem{cor}[thm]{Corollary}
\newtheorem{lemma}[thm]{Lemma}

\newtheorem{rem}[thm]{Remark}
\newcommand{\berem}{\begin{rem}$\!\!${\bf }$\;$\rm }
\newcommand{\enrem}{ \end{rem}}
\newcommand{\ip}[2]{\langle {#1}|{#2}\rangle}

\newcommand{\LD}{{\L}\ad(\D)}
\newcommand{\LDH}{{\L}\ad(\D,\H)}

\newcommand{\BH}{{\mc B}(\H)}
\newcommand{\cu}{{\rm c}}
\newcommand{\bic}[1]{{#1}''_{\w\sigma}}
\newcommand{\bicc}[1]{{#1}''_{\w\cu}}
\newcommand{\wcom}{\M'_\w\D\subset \D}
\def\NG{{\mathfrak N}}
\begin{document}
\title[ ]
{Induced and reduced unbounded operator algebras}

\author{F. Bagarello}
\address{Dipartimento di Metodi e Modelli Matematici,
Fac. Ingegneria, Universit\`a di Palermo, I-90128  Palermo, Italy}
\email{bagarell@unipa.it}

\author{A. Inoue}
\address{
Department of Applied Mathematics, Fukuoka University, Fukuoka
814-0180, Japan}
\email{a-inoue@fukuoka-u.ac.jp}
\author{C. Trapani} \address{Dipartimento di Matematica ed
Applicazioni, Universit\`a di Palermo, I-90123 Palermo
Italy}
\email{
trapani@unipa.it}

\begin{abstract}
The induction and reduction precesses of an O*-vector space $\M$ obtained by means of a projection taken, respectively, in $\M$ itself or in its weak bounded commutant $\M'_\w$ are studied. In the case where $\M$ is a partial GW*-algebra, sufficient conditions are given for the induced and the reduced spaces to be partial GW*-algebras again.
\end{abstract}

\maketitle


\section{Introduction and preliminaries} Among the "elementary'' operations one can perform on a bounded operator algebra $\M$ there are the so-called processes of {\em induction} and  {\em reduction}: both of them are obtained via a projection chosen either in $\M$ or in its commutant $\M'$ \cite{dixmier}. If $\M$ is a von Neumann algebra then both procedures give rise once more to von Neumann algebras.

In a recent paper \cite{bit1} we made some steps toward the generalization of the reduction process to unbounded operator algebras (O*-algebras). In that case there is not a strict analog of a von Neumann algebra, but EW*-algebras and  GW*-algebras behave sufficiently well to be considered as natural extensions of that notion to O*-algebras. In particular, in \cite{bit1} we gave conditions for the reduction of a GW*-algebra to be again a GW*-algebra.

In this paper we make a step further and consider the induction and reduction processes in the more general case where $\M$ is an O*-vector space on a given domain $\D$ in Hilbert space $\H$ (we maintain, in this way, our study as much general as possible)  with a special attention to the case where $\M$ is a { partial} O*-algebra or a partial GW*-algebra.

In Section \ref{sect_2} we examine the induction procedure starting from an O*-vector space $\M$ and a projection $E\in \M$. The result of this analysis is a sufficient condition for the induced space $\M_E$ of a (partial) GW*-algebra $\M$ to be a (partial) GW*-algebra.

In Section \ref{sect_3} we consider the reduction procedure moving, this time, from an O*-vector space $\M$ and a projection $E\in \M'_\w$, the weak bounded commutant of $\M$. Also in this case the main outcome is a number of sufficient conditions for  the induced space $\M_E$ of a partial GW*-algebra $\M$ to be a partial GW*-algebra.

Finally, in Section \ref{sect_4}, we show how the results obtained in Section \ref{sect_3} can be used for analyzing the existence of conditional expectations for partial O*-algebras, applying some proposition proved by Takakura \cite{takakura}.

\medskip
For reader's  convenience, we  recall here briefly the definitions of (partial) O*-algebras, (partial) GW*-algebras and other basics. More details can be found in \cite{schmu,ait_book}.

Let  $\D$ be a dense subspace of a Hilbert space $\Hil$. We denote by $ \Lc^\dagger(\D,\Hil) $ the set of all
(closable) linear operators $X$ such that $ {\D}(X) = {\D},\; {\D}(X^*) \supseteq {\D}.$

The set $ \LL^\dagger(\D,\Hil ) $ is a  partial *-algebra
 with respect to the following operations: the usual sum $X_1 + X_2 $,
the scalar multiplication $\lambda X$, the involution $ X \mapsto X^\dagger = X^* \restr{\D}$ and the \emph{
(weak)} partial multiplication $X_1 \mult X_2 = {{X_1}^\dagger}^* X_2$, defined whenever $X_2$ is a weak right
multiplier of $X_1$ (we shall write $X_2 \in R^{\rm w}(X_1)$ or $X_1 \in L^{\rm w}(X_2)$), that is, iff $ X_2
{\D} \subset {\D}({{X_1}^\dagger}^*)$ and  $ X_1^* {\D} \subset {\D}(X_2^*).$

Given a subset $\M$ of $\LDH$, its universal right multipliers are the elements
of the set:
$$R^\w\M = \{Y \in \LDH; Y \in R^\w(X), \forall X \in  \M\}.$$

Let $\Lc^\dagger(\D)$ be the subspace of $\Lc^\dagger(\D,\Hil)$ consisting of all its elements  which leave, together with their adjoints, the domain $\D$ invariant. Then $\Lc^\dagger(\D)$ is a *-algebra with
respect to the usual operations.

A $\ad$-invariant subsbace $\M$ of $\LDH$ is called an O*-vector space in $\LDH$. In particular, if  $\M\subset\LD$, then $\M$ will be called an O*-vector space in $\LD $.

An O*-vector space $\M$ is called a  (weak) partial O*-algebra on $\D$ in $\H$ if it is stable under the weak multiplication $\mult$ (in the sense that $X,\,Y \in \M$ and $Y\in R^\w(X)$ imply $X\mult Y \in \M$).

Let $\M $ be an O*-vector space in $\LDH$. The {\em graph topology} $t_\M $ on $\D$ is the locally convex topology
defined by the family $\{\|\cdot\|,  \|\cdot\|_X;\, X \in \M \}$ of seminorms: $\|\xi\|_X= \|X\xi\|$, $\xi \in \D$. If the
locally convex space $\D[t_\M ]$ is complete, then $\M $ is said to be {\em closed}. More in general, we denote
by $\widetilde{\D}(\M )$ the completion of the locally convex space $\D[t_\M ]$ and put
$$ \widetilde{X}:= \overline{X}\restr{\widetilde{\D}(\M ) }\quad\mbox{ and } \widetilde{\M }:=\{\widetilde{X}: X \in
\M \}.$$ Then $\widetilde{\M }$ is a closed
O*-vector space on $\widetilde{\D}(\M )$ which is called the {\em closure} of $\M $, since it is the smallest closed
extension of $\M $.

If $\M$ is an O*-vector space, we also put $\widehat{\D}(\M ) = \bigcap_{X\in
\M }D(\overline{X})$ and
$$ \widehat{X}:= \overline{X}\restr{\widehat{\D}(\M ) } \quad\mbox{ and } \widehat{\M }:=\{\widehat{X}: X \in
\M \}.$$
Then $\widehat{\M }$ is an O*-vector space on $\widehat{\D}(\M )$, called the {\em full closure} of $\M$.
If $\D=\widehat{\D}(\M )$ and, consequently, $\M=\widehat{\M }$, $\M$ is said to be {\em fully closed} . If $\M$ is an O*-algebra, the notions of closure and full closure coincide.

Let again $\M$ be an O*-vector space and let $\D^*(\M ) := \bigcap_{X\in
\M }D(X^*)$.  If $\D=\D^*(\M )$ then $\M $ is said to be {\em self-adjoint}.  It is clear that
\begin{align*}
& \D\subset \widetilde{\D}(\M ) \subset \widehat{\D}(\M ) \subset \D^*(\M )\\
& X\subset \widetilde{X} \subset \widehat{X}\subset {X\ad}^*, \quad \forall X \in \M .
\end{align*}
The {\em weak commutant} $\M '_\w$ of $\M $ is defined by
$$\M '_\w=\{C\in \BH: \ip{CX\xi}{\eta}=\ip{C\xi}{X\ad\eta}, \; \forall X\in \M ,\,\forall \xi, \eta\in \D \},$$
where $\BH$ is the *-algebra of all bounded linear operators on $\H$. Then $\M '_\w$ is a weak-operator closed
*-invariant subspace of $\BH$ (but it is not, in general, a von Neumann algebra) and $(\widetilde{\M })'_\w
=\M '_\w$. If $\M '_\w\D\subset \D$, as it happens for self-adjoint $\M $,  then $\M '_\w$ is a von Neumann
algebra.

In this paper we will also use unbounded commutants and bicommutants of $\M $, whose definitions we shortly recall.

The {\em weak unbounded commutant} $\M '_\sigma$ of $\M $ is defined by
$$\M '_\sigma=\{Y\in \LDH: \ip{X\xi}{Y\ad\eta}=\ip{Y\xi}{X\ad\eta}, \; \forall X\in \M ,\,\forall \xi, \eta\in \D \}.$$

We put $\M '_\cu := \M '_\sigma \cap \LD$.
If $\M$ is an O*-vector space, so is also $\M '_\sigma$.
As for the bicommutants, the bounded one
$$(\M '_\w)' =\{A \in \BH: \, AX=XA,\, \forall X \in \M '_\w\}$$
is a von Neumann algebra on $\H$. The unbounded bicommutant is defined as
$$\bic{\M }:=\{X\in \LDH:\, \ip{CX\xi}{\eta}=\ip{C\xi}{X\ad\eta}, \;
\forall C\in \M '_\w,\,\forall \xi, \eta\in \D \}.$$ Then,
$\bic{\M }$ is a $\tau_{s^*}$-closed partial O*-algebra on $\D$ such that $\M  \subset \bic{\M }$,
where the strong*-topology $\tau_{s^*}$ is defined by the family $\{p_\xi^*(\cdot); \xi \in \D\}$ of
seminorms:
$$p_\xi^*(X):=\|X\xi\|+\|X\ad\xi\|,\quad X \in \LDH.$$
Furthermore, if $\M '_\w\D\subset \D$, then $\bic{\M}$ is a $\tau_{s^*}$-closed partial O*-algebra on $\D$ and
$$\bic{\M } =\{X \in \LDH:\, \overline{X} \mbox{ is affiliated with }(\M '_\w)'\}=
\overline{(\M '_\w)'\restr{\D}}^{\tau_{s^*}}. $$
Similarly, $\bicc{\M}$ is a $\tau_{s^*}$-closed O*-algebra on $\D$. One has $(\bicc{\M})'_\w = \M '_\w$ and
$$\bicc{\M } =\{X \in \LD:\, \overline{X} \mbox{ is affiliated with }(\M '_\w)'\}=
\overline{(\M '_\w)'\restr{\D}}^{\tau_{s^*}}\cap \LD. $$

A fully closed partial O*-algebra $\M$ on $\D$ is called a {\em partial
GW*-algebra} if $\wcom$ and $\M = \bic{\M}$. A closed O*-algebra $\M$ on $\D$ is called
a {\em GW*-algebra} if $\wcom$   and $\M = \bicc{\M}$.

Finally, if $\M \subset \LDH$, we denote by $\M_b$ the {\em bounded part} of $\M$ (i.e., the subset of bounded operators of $\M$) and put  $\overline{\M_b}=\{\overline{X}; X \in \M_b\}$.   A fully closed (partial) O*-algebra $\M$ on $\D$ is said to be
a {\em (partial) EW*-algebra} if $\overline{\M_b}$ is a von Neumann algebra, $\M_b \D  \subset \D$ and $\overline{X} $is affiliated with $\overline{\M_b}$ for every $X\in \M$.

\section{Induced GW*-algebras and induced partial GW*-algebras} \label{sect_2}

Let $\M$ be a (fully closed) O*-vector space in $ \LL^\dagger(\D,\Hil ) $ and $E$ a projection in $\M$ such that $E\D\subset\D$. Then we put
$$
\M_E=\{EX\restr{E\D}:\,X\in\M\}.
$$
Then $\M_E$ is an O*-vector space in $ \LL^\dagger(E\D,E\Hil ) $, called the {\em induction} of $\M$. We can check that if $\M$ is a O*-vector space in $ \LL^\dagger(\D) $ then $\M_E$ is a O*-vector space in $ \LL^\dagger(E\D) $ and that if $\M$ is a (partial) O*-algebra on $ \D $ then $\M_E$ is a (partial) O*-algebra on $ E\D $.

\berem The full-closability of $\M$  does not automatically imply the full-closability of $\M_E$, and $(\M_E)'_\w$ is not necessarily a von Neumann algebra.
\enrem
In the framework considered here it is natural to answer to the following questions:

{\bf Question 1:} Let $\M$ be a partial GW*-algebra on $\D$. Is $\M_E$ a partial GW*-algebra on $E\D$?

{\bf Question 2:} Let $\M$ be a GW*-algebra on $\D$. Is $\M_E$ a GW*-algebra on $E\D$?

To answer these questions we first give the following

\begin{lemma}\label{lemma 2.2}
Let $\M$ be a O*-vector space in $ \LL^\dagger(\D,\Hil ) $ such that $\M'_\w\D\subset\D$ and $E$ a projection in $\M$ such that $E\D\subset\D$. Assume also that the closure of $\M_b$ with respect to the strong *-topology, $\overline{\M_b}^{s^*}$, coincides with $(\M_\w')'$. Then
$$
(\M_\w')_E=(\M_E)'_\w,\qquad (\bic{\M})_E=(\M_E)''_{\w\sigma},\qquad (\M_{\w c}'')_E=(\M_E)''_{\w c}.
$$

\end{lemma}
\begin{proof}
Since $\M'_\w\D\subset\D$, then $\M_\w'$ is a von Neumann algebra. The assumption $\overline{\M_b}^{s^*}=(\M_\w')'$, in turn, implies that
  \begin{equation}\label{3.1}(\M_E)'_\w=(((\M_\w')')_E)'=((\M_\w')'')_E = (\M_\w')_E.
 \end{equation}
The equalities \eqref{3.1} easily imply that
$(\bic{\M})_E=\bic{(\M_E)}$ and $(\bicc{\M})_E=\bicc{(\M_E)}$.
\end{proof}
Every partial GW*-algebra satisfies the assumptions of Lemma \ref{lemma 2.2}. This is not necessarily true for a GW*-algebra. Hence we have
\begin{thm}
Let $\M$ be a  partial GW*-algebra on $\D$ and $E$ a projection in $\M$ such that $E\D \subset \D$. Then the full closure of $\M_E$ is a partial GW*-algebra on $\hat\D(\M_E)$.
\end{thm}
\begin{thm}\label{2.4}
If  $\M$ is a   GW*-algebra on $\D$  such that $\overline{\M_b}^{s^*}=(\M_\w')'$ and $E$ is a projection in $\M$, then the  closure of $\M_E$ is a  GW*-algebra on $\tilde\D(\M_E)$.
\end{thm}

\begin{cor} Let $T$ be a self-adjoint operator in $\H$ and $\M$ be a GW*-algebra on $\D^\infty(T)= \bigcap_{n\in {\mb N}} D(T^n)$, containing $T\restr \D^\infty(T)$. Let $E$ be a projection in $\M$. Then the closure of $\M_E$ is a GW*-algebra.
\end{cor}
\begin{proof} Let $T=\int_{-\infty}^ \infty \lambda dE_T(\lambda)$ be the spectral resolution of $T$. Put $E_k=E_T(k)-E_T(-k)$, $k \in {\mb N}$. Then it is easily seen that the set $\{E_kAE_\ell;\, A\in (\M_\w')',\,k,\ell \in {\mb N}\}$ is contained in $\overline{\M_b}$ and it is dense in $(\M_\w')'$ with respect to the strong*-topology. This implies, by Theorem \ref{2.4}, that the closure of $\M_E$ is a GW*-algebra.
\end{proof}
\section{Reduced partial GW*-algebras} \label{sect_3}

In a previous paper \cite{bit1} we have studied the so-called {\em reduced O*-algebras} obtained via a projection picked in the weak commutant. Here we consider the analogous problem for partial O*-algebras and we show that similar results can still be deduced.

Let $\M$ be a fully closed O*-vector space in $ \LL^\dagger(\D,\Hil ) $ such that $\M_\w'\D\subset\D$, and $E'$ a projection in $\M_\w'$. We put
$$
\M_{E'}=\{X_{E'}:=X\restr{E'\D}:\,X\in\M\}.
$$
Then $\M_{E'}$ is a fully closed O*-vector space on $E'\D$ satisfying the following properties:
\begin{itemize}

\item $(\M_{E'})'_\w=(\M_\w')_{E'}$,
\item $(\M_{E'})'_\w E'\D\subset E'\D$,
\item $((\M_{E'})'_\w)'=((\M_{\w}')_{E'})'=((\M_\w')')_{E'}$,
\item if $\M$ is a O*-algebra on $\D$ then $\M_{E'}$ is a O*-algebra on $E'\D$.
\end{itemize}
However, even if $\M$ is a partial O*-algebra on $\D$, it is not guaranteed that $\M_{E'}$ is a partial O*-algebra on $E'\D$. This is because $X_{E'}\mult Y_{E'}$ could be well defined, without $X\mult Y$ being well defined.

\medskip
The set $\M_{E'}$ is the {\em reduction} of $\M$. We have the following questions:

{\bf Question 3:} Suppose that $\M$ is a partial GW*-algebra on $\D$ and $E'$ a projection in $\M_\w'$. Is then $\M_{E'}$ a partial GW*-algebra on $E'\D$?

To answer to this question it is necessary first to answer the following preliminary

{\bf Question 4:} Let $\M$ be a fully closed O*-vector space in $ \LDH$ such that $\M_\w'\D\subset\D$, and $E'$ a projection in $\M_\w'$. Does then the equality $(\bic{\M})_{E'}=\bic{(\M_{E'})}$ hold?

It is clear that $(\bic{\M})_{E'}\subset\bic{(\M_{E'})}$.

We now look for conditions under which the converse inclusion holds.

 Let $Z_{E'}$ be the central support of $E'$, that is the projection on $\overline{\left<\M_\w'E'\Hil\right>}$. Let $\E\equiv \left<\M_\w'E'\D\right>\oplus (\Id-Z_{E'})\D\subset\D$. Clearly, $\left<\M_\w'E'\D\right>\subset Z_{E'}\D$.

For  $X\in \bic{(\M_{E'})}$, we define $$X_e\left(\sum_k\,C_k\,E'\,\xi_k+(\Id-Z_{E'})\eta\right):= \sum_k\,C_k\,X\,E'\,\xi_k,$$ with $C_k\in\M_\w',$ $\xi_k,\eta\in\D.$ Then $X_e$ extends $X$ to $\E$.

It is easily shown that $e(\M_{E'})\equiv (\bic{(\M_{E'})})_e:=\{X_e: \,X\in \bic{(\M_{E'})}\}$ is a partial O*-algebra on $\E$ in $\Hil$.

\begin{lemma}\label{Lemma 4.1}
Let $\hat\D(e(\M_{E'}))=\D$ or, equivalently, let the full closure of $\left<\M_\w'E'\D\right>$ w.r.t. $e(\M_{E'})$ equal $Z_{E'}\D$. Then $(\bic{\M})_{E'}=\bic{(\M_{E'})}$.
\end{lemma}

\begin{proof} By the assumption, for every $X\in \bic{(\M_{E'})}$, $\hat{X_e}$ is an operator on $\D$ and $\hat{X_e}\restr{E'\D}= X$. Furthermore, for every $C, C_k \in \M_\w',\, \xi_k, \eta \in \D $, we have
\begin{eqnarray*} CX_e\left( \sum_k\,C_k\,E'\,\xi_k+(\Id-Z_{E'})\eta \right)&=& \sum_k\,CC_k\,XE'\,\xi_k \\
&=& X_e C\left( \sum_k\,C_k\,E'\,\xi_k+(\Id-Z_{E'})\eta \right).
\end{eqnarray*}
Therefore, $\hat{X_e} \in \bic{\M}$ and, hence, $X= (\hat{X_e})_{E'}\in {(\bic{\M})}_{E'}$.
\end{proof}

\berem We notice that, under the conditions of Lemma \ref{Lemma 4.1}, ${(\bic{\M})}_{E'}$ is a partial O*-algebra. This may fail to be true in the general case.
\enrem

The condition of Lemma \ref{Lemma 4.1} is actually satisfied in some interesting situation, as the proof of the next theorem shows.

\begin{thm} \label{thm_4.2}Let $\M$ be a self-adjoint O*-vector space on $\D$ in $\Hil$. Assume that the topology $t_\M$ is defined by a norm-increasing sequence $\{T_n\}$ (i.e., $\|T_n\xi\| \leq \|T_{n+1}\xi\|,$ $\forall \xi\in \D, n \in {\mb N}$) of essentially self-adjoint operators of $\M$, whose spectral projections leave the domain $\D$ invariant. Let $E'$ be a projection in $\M_\w'$. Then, $(\bic{\M})_{E'}=(\M_{E'})_{\w\sigma}''$. In particular, if $\M$ is a partial GW*-algebra, then $\M_{E'}$ is a partial GW*-algebra on $E'\D$.
\end{thm}

\begin{proof} Since $\M$ is a self-adjoint, then $\M_\w' \D \subset \D$ and $\bic{\M}$ is a self-adjoint partial O*-algebra on $\D$. Furthermore,
\begin{equation}\label{4.1} t_\M= t_{\bic{\M}}= t_{\Lc^\dagger(\D,\Hil)}\end{equation} since $\M$ is closed and $t_\M$ is metrizable.
Let us consider again $\E\equiv \left<\M_\w'E'\D\right>\oplus (\Id-Z_{E'})\D\subset\D$, where $Z_{E'}$ is the central support of $E'$.

Now, take an arbitrary $X \in (\M_{E'})_{\w\sigma}''$. Then, since $X=E'XE' \in \LDH$, from \eqref{4.1} it follows that
\begin{equation}\label{3.2_new}
\|XE' \xi\| \leq \gamma \| T_{n_0}E'\xi\|, \quad \forall \xi \in \D,
\end{equation}
for some $\gamma>0$ and $n_0\in {\mb N}$. Let us now consider arbitrary operators $C_k \in \M_\w'$ and vectors $\xi_k\in \D$, $k=1,2, \ldots, n$. Then, the $n \times n$ matrix $(E'C_j^*C_kE')$ is positive. Let us denote by $(B_{jk})$ its square root. Then, since $B_{jk}=E'B_{jk}E' \in (\M_\w')_{E'}$ and $X \in ((\M_\w')_{E'})'_\sigma$, we have

\begin{eqnarray}\label{newineq}
{}\;\;\left\| X_e\left( \sum_k C_k E' \xi_k\right)\right\|^2&=& \left\|\sum_k C_k X E' \xi_k\right\|^2\nonumber\\
&=& \sum_{k,j}\ip{E'C_j^*C_kE'XE'\xi_k}{XE'\xi_j} \nonumber \\
&=& \sum_j\left\|\sum_k B_{jk}XE'\xi_k \right\|^2 \nonumber \\
&=& \sum_j\left\|X\left(\sum_k B_{jk}E'\xi_k\right) \right\|^2 \\
&\leq &\gamma^2 \sum_j\left\|T_{n_0}\left(\sum_k B_{jk}E'\xi_k\right)\right\|^2 \, \mbox{(by \eqref{3.2_new})}\nonumber \\
&=& \gamma^2 \left\|\sum_kC_kT_{n_0} E'\xi_k\right\|^2 \nonumber\\
&=& \gamma^2 \left\|T_{n_0}\left( \sum_kC_k E'\xi_k \right)\right\|^2, \nonumber
\end{eqnarray}
which implies that the full closure of $\left<\M_\w'E'\D\right>$ w.r.t. $e(\M_{E'})$ equals $Z_{E'}\D$. Indeed, take an arbitrary $\xi \in \D$. Then, there exists a sequence $\{\xi_n\}$ in $\left<\M_\w'E'\D\right>$ which converges to $Z_{E'}\xi$.
 Let $
\overline{T}_{n_0}=\int_{-\infty}^{\infty} \lambda d E_{T_{n_0}}(\lambda)$ be the spectral resolution of the self-adjoint operator $\overline{T}_{n_0}$. We put $E_k:=E_{T_{n_0}}(k)-E_{T_{n_0}}(-k)$, $k \in {\mb N}$. Then, $E_k \in (\M'_\w)'$ and $E_k \D \subset \D$, for every $k \in {\mb N}$, by the assumption. Hence, $E_k\xi_n\in \left<\M_\w'E'\D\right>$, $\forall k,n \in {\mb N}$ and $$\lim_{k\to\infty}\lim_{n\to\infty}E_k\xi_n=Z_{E'}\xi.$$
Furthermore, since
$$\lim_{k\to\infty}\lim_{n\to\infty}T_{n_0}E_k\xi_n=\lim_{k\to\infty}T_{n_0}E_kZ_{E'}\xi= T_{n_0}Z_{E'}\xi,$$ from \eqref{newineq} it follows that $Z_{E'}\xi \in D(\overline{X_e})$ and $\overline{X_e}Z_{E'}\xi= \lim_{k,n \to \infty}X_e E_k \xi_n$, which implies that $Z_{E'}\xi \in \widehat{\D}(e(\M_{E'}))$. Thus, the full closure of $\left<\M_\w'E'\D\right>$ w.r.t. $e(\M_{E'})$ equals $Z_{E'}\D$. \\
By Lemma \ref{Lemma 4.1} we have
$$\bic{(\M_{E'})} = (\bic{\M})_{E'}.$$
Therefore, if $\M$ is a partial GW*-algebra, then $\M_{E'}$ is also a partial GW*-algebra.
\end{proof}

\berem In the proof of Theorem 2.5 of \cite{bit1} we claimed that the inequality
\begin{equation}\label{gap} \| X_e (CE'\xi)\|\leq \gamma \|T_{n_0}CE' \xi \|, \, \forall C \in \M'_\w, \forall \xi \in \D\end{equation} (which was proved in that paper) implies the condition (i) of Lemma 2.3 of \cite{bit1}. Actually, \eqref{gap} is not enough. Indeed what we needed to prove was, in fact, that
 $$\left\| X_e \left(\sum_kC_kE'\xi_k\right)\right\|\leq \gamma \left\|T_{n_0}\left(\sum_kC_kE'\xi_k\right)\right \|, \, \forall C_k \in \M'_\w, \forall \xi_k \in \D.$$
This inequality does really hold: it can be proved similarly to \eqref{newineq}.

\enrem
By Theorem \ref{thm_4.2} we get the following

\begin{cor} \label{cor_4.3} Let $T$ be a self-adjoint operator in $\H$ and  $\M$ be an O*-vector space in $\L\ad(\D^\infty (T),\H)$ where $\D^\infty (T):=\bigcap_{n\in {\mb
N}}D(T^n)$, containing $T^n\restr \D^\infty (T)$, for every $n \in {\mb N}$. Let $E'$ be a projection in $\M'_\w$. Then
$\bic{(\M_{E'})}=(\bic{\M})_{E'}$. In particular, if $\M$ is a partial GW*-algebra on $\D^\infty (T)$, then $\M_{E'}$ is a partial GW*-algebra on $E'\D^\infty (T)$.
\end{cor}
\begin{cor} \label{cor_4.4} Let $\M$ be a partial EW*-algebra on the Fr\'echet domain $\D$ in Hilbert space $\H$ and
$E'$ a projection in $\M'_\w$. Then, $\bic{(\M_{E'})}=(\bic{\M})_{E'}$.
\end{cor}

\section{Applications} \label{sect_4}

In this section we show how to use the results of Section \ref{sect_3} in the analysis of the existence of conditional
expectations for partial O*-algebras,  studied by Takakura \cite{takakura}.

\begin{thm}\label{thm_5.1}
Let $\M$ be a self-adjoint partial GW*-algebra on $\D$ in $\H$ and let $\xi_0$ be a strongly cyclic (i.e., $(R^\w\M)\xi_0$ is dense in $\D[t_\M]$)
and separating vector, in the sense that $\overline{\M'_\w\xi_0}=\H$. Let $\NG$ be a partial GW*-subalgebra of $\M$ satisfying
\begin{itemize}

\item[(N$_0$)] $(\NG\cap R^\w\M)\xi_0$ is dense in $\H_\NG:= \overline{\NG\xi_0}$;

\item[(N$_1$)] $\NG'_\w \D \subset \D$ ;
\item[(N$_2$)] $(\NG\cap R^\w\M)\xi_0$ is essentially self-adjoint for $\NG$

\item[(N$_3$)] ${\Delta_{\xi_0}''}^{it}(\NG_\w')'{\Delta_{\xi_0}''}^{-it}=(\NG_\w')'$, $\forall t\in\mathbb{R}$, where ${\Delta_{\xi_0}''}$ is the modular operator of the full left Hilbert algebra $(\M_\w')'\xi_0$.
\end{itemize}
Suppose that $t_\NG$ is defined by a sequence $\{T_n\}$ of essentially self-adjoint operators of $\NG$ whose spectral projections leave the domain $\D$ invariant. Then there exists a unique conditional expectation $\E$ of $(\M, \xi_0)$ with respect to $\NG$; that is $\E$ is a map of $\M$ onto $\NG$ such that
\begin{itemize}

\item[(i)] $\E(A)^\dagger=\E(A^\dagger)$,  $\forall\,A\in\M$,  and $\E(X)=X$, $\forall X\in\NG$;

\item[(ii)] $\E(A\mult X)=\E(A) \mult X$, for all $A\in \M$ and $X\in \NG\cap R^\w\M$ and $\E(X\mult A)=X \mult\E(A)$,  for all $A\in R^\w\M $, and $X\in\NG$;

\item[(iii)] $\omega_{\xi_0}(\E(A)):=\ip{\E(A)\xi_0}{\xi_0}=\ip{A\xi_0}{\xi_0}=\omega_{\xi_0}(A)$, for all $\,A\in\M$.
\end{itemize}
\end{thm}

\begin{proof} Since $\M$ is self-adjoint and $(\NG\cap R^\w\M)\xi_0$ is dense in $\H_\NG$, it follows that $(\NG\cap R^\w\M)\xi_0$ is a reducing subspace for $\NG$; that is
$$ \NG (\NG\cap R^\w\M)\xi_0\subset \overline{(\NG\cap R^\w\M)\xi_0}=\H_\NG,$$
which implies by (N$_2$) and \cite[Theorem 7.4.4]{ait_book} that $P_\NG \in \NG'_\w$ and $P_\NG \D \subset \D$, where $P_\NG$ is the projection of $\H$ onto $\H_\NG$. Moreover, from Theorem \ref{thm_4.2} it follows that $\NG_{P_\NG}$ is a partial GW*-algebra on $P_\NG \D$. Now, define $\E_\NG(A):= P_\NG A \restr{P_\NG D}$, $A\in \M$. Then by \cite[Theorem 5.3]{takakura} it follows that $\E_\NG$ is a conditional expectation of $(\M, \xi_0)$ with respect to $\NG$ and that it is unique. This concludes the proof.
\end{proof}

\begin{cor} Let $\M$, $\xi_0$, $\NG$ be as in Theorem \ref{thm_5.1}. Suppose that $\D=\D^\infty (T)$, where $T$ is a  self-adjoint operator of $\NG$ such that $T^n \restr \D^\infty (T) \in \NG$, for every $n \in {\mb N}$. Then  there exist a unique  conditional expectation $\E_\NG$ of $(\M, \xi_0)$ with respect to $\NG$.
\end{cor}
\begin{proof} This follows from Corollary \ref{cor_4.3} and Theorem \ref{thm_5.1}.
\end{proof}

\section*{Acknowledgements}
The authors wish to express their gratitude to the referee for pointing out a mistake in a earlier version of the proof of Theorem \ref{thm_4.2} and for his suggestions.

This work has been done during mutual visits of the authors to their respective home institutions:  F.B. and C.T. acknowledge
the warm hospitality of the Department of Applied Mathematics of the Fukuoka University as well as A.I. acknowledges the hospitality of the Dipartimento di Matematica
ed Applicazioni, Universit\`a di Palermo. We also acknowledge grants of
 the Japan Private School Promotion Foundation and  of CORI, Universit\`a di Palermo.

\end{document}